\documentclass{article}
\usepackage[numbers]{natbib}

\usepackage{graphicx} % Required for inserting images

\usepackage{amsthm}
\newtheorem{theorem}{Theorem}
\usepackage[left=0.52in,right=0.52in,top=0.3in,bottom=0.3in,includeheadfoot]{geometry} 

\usepackage{appendix}
\usepackage{float} 
\usepackage{amsmath,epsfig,pstricks,pst-text,pst-node,pst-plot,amssymb,wrapfig,threeparttable,rotating}

\usepackage{animate}
\usepackage{booktabs}
\makeatother
\title{Extending the Accelerated Failure Conditionals Model to Location-Scale Families%
\thanks{For full functionality of the animated figures presented in this document, please view this PDF in \textbf{Adobe Acrobat Reader}. Other PDF viewers (including browser-based viewers, 
Preview) may render the figures statically.}}
\author{
  Jared N. Lakhani\\
  \textit{Department of Statistical Sciences, University of Cape Town}\\
  \texttt{lkhjar001@myuct.ac.za}
}
\date{} % You can also write a fixed date or leave blank 
\usepackage[hidelinks]{hyperref}

\begin{document}

\maketitle
\begin{abstract}
\citet{arnold2020bivariate} introduced a novel class of bivariate conditionally specified distributions, in which dependence between two random variables is established by defining the distribution of one variable conditional on the other. This conditioning regime was formulated through survival functions and termed the accelerated failure conditionals model. Subsequently, \citet{lakhani2025modelsacceleratedfailureconditionals} extended this conditioning framework to encompass distributional families whose marginal densities may exhibit unimodality and skewness, thereby moving beyond families with non-increasing densities. The present study builds on this line of work by proposing a conditional survival specification derived from a location-scale distributional family, where the dependence between $X$ and $Y$ arises not only through the acceleration function but also via a location function. An illustrative example of this new specification is developed using a Weibull marginal for $X$. The resulting models are fully characterized by closed-form expressions for their moments, and simulations are implemented using the Metropolis-Hastings algorithm. Finally, the model is applied to a dataset in which the empirical distribution of $Y$ lies on the real line, demonstrating the models' capacity to accommodate $Y$ marginals defined over $\mathbb{R}$.
\end{abstract}
\section{Introduction}
We extend the accelerated failure conditionals model, originally introduced by \citet{arnold2020bivariate}, to the real line by re-specifying the conditional survival function to have support on $y \in (-\infty, \infty)$. Specifically, we consider a heterogeneous family specification in which $X$ belongs to some positive support distributional family, while $Y$ belongs to a location-scale family. Accordingly, the conditional survival function is redefined to conform to this specification, with $\mu(x)$ denoting the location parameter and $\beta(x)$ the corresponding scale parameter. We note that \citet{arnold2020bivariate} originally introduced the model in which $\beta(x)$ acted as a rate parameter within the conditional survival function. The resulting general accelerated failure conditionals model is expressed as:
\begin{align}
    \bar{F}_X(x) = P(X>x) = \bar{F}_0(x),\quad x>0, \label{eq: general marginal}
\end{align}
for some survival function $\bar{F}_0(x) \in [0, 1]$ for $x >0$, and for each $x >0$:
\begin{equation}
    P(Y>y\mid X>x) = \bar{F}_1\left(\frac{y - \mu(x)}{\beta(x)}\right), \quad x>0,\; y \in \mathbb{R},\label{eq: general cond}
\end{equation}
for some survival function $\bar{F}_1\left(\frac{y - \mu(x)}{\beta(x)}\right) \in [0, 1]$ for $x>0$ and  $y \in \mathbb{R}$, a suitable acceleration (or, perhaps more aptly, deceleration) function $\beta(x) \in \mathbb{R} \setminus \{0\}$ and a suitable location function $\mu(x) \in \mathbb{R}$. We call $\beta(x)$ an acceleration function because it scales the argument of the baseline survival $\bar{F}_1$, thereby accelerating or decelerating the rate of decay of the conditional survival $P(Y>y\mid X>x)$. By contrast, we call $\mu(x)$ the location function because it shifts the argument of the baseline survival $\bar{F}_1$. As discussed by \citet{arnold2020bivariate}, in the analysis of dependent lifetimes of components in a system, it is more appropriate to consider the conditional density of $Y$ given that the first component, with lifetime $X$, remains operational at time $x$. Hence why, the conditioning event is taken to be $\{X > x\}$, rather than conditioning on the exact value $X = x$. \\

The joint survival function is given by:
\begin{align}
    P(X>x, Y>y) = \bar{F}_0(x) \bar{F}_1\left(\frac{y - \mu(x)}{\beta(x)}\right), \quad x >0, \; y \in \mathbb{R}.\label{eq: general joint}
\end{align}
Assuming differentiability and $X,Y$ are continuous, we obtain the marginal densities:
\[
f_X(x) = -\frac{d}{dx}P(X>x) = -\frac{d}{dx}\bar{F}_0(x) = f_0(x),
\]
and, since $\lim_{x\to 0^+}\bar{F}_0(x)=1$, where $\beta(0):=\lim_{x\to0^+}\beta(x)$ and $\mu(0):=\lim_{x\to 0^+}\mu(x)$ :
\[
P(Y>y) = \lim_{x\to 0^+}P(Y>y,X>x) = \bar{F}_1\left(\frac{y-\mu(0)}{\beta(0)}\right).
\]
Now, the density of $Y$ is:
\[
f_Y(y) = -\frac{d}{dy}P(Y>y) 
= -\frac{d}{dy}\bar{F}_1\left(\frac{y-\mu(0)}{\beta(0)}\right)
= \frac{1}{\beta(0)}\,f_1\left(\frac{y-\mu(0)}{\beta(0)}\right),
\]
where $f_1(t)=-\bar{F}_1'(t)$ denotes the density associated with $\bar{F}_1$. Hence the marginal distribution of $Y$ belongs to the same family as $\bar{F}_1$, but with its argument scaled by $\frac{1}{\beta(0)}$. For \eqref{eq: general joint} to be a valid survival function, it must have a non-negative mixed partial derivative. Denoting $f_0$ and $f_1$ as the densities corresponding to the survival function $\bar{F}_0$ and $\bar{F}_1$, and differentiating, we obtain:
\begin{align}
    \frac{\partial}{\partial x} \frac{\partial}{\partial y} P(X>x, Y>y)  &= 
f_0(x) f_1\left( \frac{y-\mu(x)}{\beta(x)}\right) \frac{1}{\beta(x)} 
 - \bar{F}_0(x) f_1'\left( \frac{y-\mu(x)}{\beta(x)}\right) \frac{1}{\beta(x)} \Bigg[ \frac{-\mu'(x)\beta(x) - \beta'(x)(y- \mu(x))}{\beta(x)^2} \Bigg] \nonumber \\
&+\bar{F}_0(x)  f_1\left( \frac{y-\mu(x)}{\beta(x)}\right)\frac{ \beta'(x) }{ \beta(x)^2}\nonumber \\
&=f_0(x) f_1\left( \frac{y-\mu(x)}{\beta(x)}\right) \frac{1}{\beta(x)} 
- \frac{\bar{F}_0(x)}{\beta(x)^2}  \Bigg[ -f_1'\left( \frac{y-\mu(x)}{\beta(x)}\right) \mu'(x) -\beta'(x)\Bigg( \frac{y-\mu(x)}{\beta(x)}f_1'\left( \frac{y-\mu(x)}{\beta(x)}\right) \nonumber\\
&+  f_1\left( \frac{y-\mu(x)}{\beta(x)}\right) \Bigg)     \Bigg]
 . \nonumber
\end{align}
Denoting $t = \frac{y-\mu(x)}{\beta(x)}$, we have:
\begin{align}
    \frac{\partial}{\partial x} \frac{\partial}{\partial y} P(X>x, Y>y) 
&= f_0(x) f_1\left( t\right) \frac{1}{\beta(x)} 
 - \frac{\bar{F}_0(x)}{\beta(x)^2} \Big[-f_1'\left( t\right)\mu'(x)  - \beta'(x) \big[tf_1'\left( t\right)
+ f_1\left( t\right) \big] \Big] \nonumber \\
& =  f_0(x) f_1\left( t\right) \frac{1}{\beta(x)} 
 +\frac{\bar{F}_0(x)}{\beta(x)^2} \Big[f_1'\left( t\right)\mu'(x)  + \beta'(x) \big[tf_1\left( t\right)\big]' \Big] \label{eq: t bounds}
\end{align}
To obtain interpretable bounds for $\beta'(x)$ and $\mu'(x)$ in \eqref{eq: t bounds}, we proceed by isolating their respective effects: first by holding $\mu(x)$ constant while allowing $\beta(x)$ to vary with $x$, and then by holding $\beta(x)$ constant while allowing $\mu(x)$ to vary with $x$. 
With the first case, given a fixed location, $\mu(x) = \mu$, that is $\mu'(x) = 0$, a necessary and sufficient condition for $   \frac{\partial}{\partial x} \frac{\partial}{\partial y} P(X>x, Y>y)  \geq 0$,  is:
\[
\beta'(x)\;\begin{cases}
\ge\; -\dfrac{f_0(x)}{\bar F_0(x)}\;\dfrac{f_1(t)}{[t f_1(t)]'}\;\beta(x), & \text{if } [t f_1(t)]'>0,\\[10pt]
\le\; -\dfrac{f_0(x)}{\bar F_0(x)}\;\dfrac{f_1(t)}{[t f_1(t)]'}\;\beta(x), & \text{if } [t f_1(t)]'<0, \\
\text{unconstrained}, & \text{if } [t f_1(t)]'=0.
\end{cases}
\]
Denoting the hazard function $h_0(x) = \frac{f_0(x)}{\bar{F}_0(x)}$ and $S(t) = \frac{f_1(t)}{[t f_1(t)]'}$, we may write the bounds as such:
\begin{align}
\sup_{t: [t f_1(t)]'>0}  -S(t) h_0(x)\beta(x) 
\;\leq\; \beta'(x) 
\;\leq\;
\inf_{t: [t f_1(t)]'<0} -S(t)h_0(x)\beta(x). \nonumber
% \label{eq: B'(x) bounds}
\end{align}
The role of the acceleration function, $\beta(x)$, in determining the sign of the correlation between $X$ and $Y$ is formalised in Theorem \ref{thm: cov(x, y) symmetric} (Appendix \ref{app: gen theorems}). In particular, since several common location-scale distributions are symmetric about their location parameter $\mu$, the theorem shows that such symmetry yields zero covariance between $X$ and $Y$. Although this does not preclude nonlinear forms of dependence - such as scale or tail dependence - it implies the absence of linear association between the variables. Consequently, the study excludes the fixed location case, as such models would capture only scale-based dependence while exhibiting zero linear correlation, limiting their usefulness for characterising the direction of dependence.\\

Hence, we consider the second case from \eqref{eq: t bounds}; given constant acceleration, $\beta(x) = \beta$, that is $\beta'(x) = 0$, a necessary and sufficient condition for $  \frac{\partial}{\partial x} \frac{\partial}{\partial y} P(X>x, Y>y)  \geq 0$,  is:
\[
\mu'(x)\;\begin{cases}
\ge\; -\dfrac{f_0(x)}{\bar F_0(x)}\;\dfrac{f_1(t)}{ f_1'(t)}\;\beta, & \text{if } f_1'(t)>0,\\[10pt]
\le\; -\dfrac{f_0(x)}{\bar F_0(x)}\;\dfrac{f_1(t)}{ f_1'(t)}\;\beta, & \text{if } f_1'(t)<0, \\
\text{unconstrained}, & \text{if } f_1'(t)=0.
\end{cases}
\]
Denoting the hazard function $h_0(x) = \frac{f_0(x)}{\bar{F}_0(x)}$ and $Q(t) = \frac{f_1(t)}{f_1'(t)}$, we may write the bounds as such:
\begin{align}
\sup_{t: f_1'(t)>0}  -Q(t) h_0(x)\beta
\;\leq\; \mu'(x) 
\;\leq\;
\inf_{t:  f_1'(t)<0} -Q(t)h_0(x)\beta. \label{eq: mu'(x) bounds}
\end{align}
The role of the location function $\mu(x)$ in determining the sign of the correlation between $X$ and $Y$ is established in Theorem~\ref{thm: cov(x, y)} (Appendix \ref{app: gen theorems}). Accordingly, the present study uses $\mu^{+}(x)$ to denote the choice of location function that arises from the extreme case in which $\mu'(x)$ attains its theoretical upper bound, hence only allowing for modelling positively correlated $X$ and $Y$. Conversely, $\mu^{-}(x)$ denotes the location function arising from the extreme case in which $\mu'(x)$ attains its theoretical lower bound, hence only allowing for modelling negatively correlated $X$ and $Y$. Now given $\inf_{t:f_1'(t)<0}(-Q(t)) = s^* > 0$, from bounds in \eqref{eq: mu'(x) bounds}:
\begin{align}
    \frac{d\left(\mu^+(x)\right)}{dx} &= s^*h_0(x) \beta, \nonumber \\
    \therefore \int d\left(\mu^+(x)\right) &=  s^*\beta\int  h_0(x) dx, \nonumber\\
    \therefore  \mu^+(x)  &= \gamma - s^*\beta \log\left( \bar{F}_0 (x)\right), \label{eq: mu+(x) general}
    \end{align}
for $\gamma \in \mathbb{R}$. This implies $\mu^+(x)$ is non-decreasing for $s^*>0$. Furthermore, $\mu(0): = \lim_{x\to 0^+} \mu(x) = \gamma$ since $\lim_{x\to 0^+}\bar{F}_0(x) = 1$. We also note that since $\bar{F}_0(x)$ approaches zero only as $x\rightarrow \infty$, the expression in \eqref{eq: mu+(x) general} remains non-zero for all finite $x$. Similarly, if we denote $\sup_{t:f_1'(t)>0}(-Q(t)) = c^* <0$, we have:
\begin{align}
    \mu^{-}(x)  = \gamma - c^*\beta \log\left( \bar{F}_0 (x)\right), \label{eq: mu-(x) general}
\end{align}
for $\gamma \in \mathbb{R}$ and implies $\mu^-(x)$ is non-increasing for $c^*<0$.\\

Furthermore, as done in \citet{arnold2020bivariate}, we introduce a dependence parameter $\tau$ which controls the strength of the dependence between $X$ and $Y$. Extending from \eqref{eq: mu+(x) general} and \eqref{eq: mu-(x) general}, we define the location functions as such:
\begin{align}
    \mu^+(x) &= \gamma - s^*\beta \log\left( \bar{F}_0 (\tau x)\right),\label{eq: mu+(x) general 2} \\
    \mu^-(x) &= \gamma - c^*\beta \log\left( \bar{F}_0 (\tau x)\right), \label{eq: mu-(x) general 2} 
\end{align}
for $\gamma \in \mathbb{R}$, $s^*>0$ and $c^*<0$. Hence, the parameter $\tau$ modulates the rate at which $\mu(x)$ varies with $x$, thereby serving as a dependence-scaling parameter. When $\tau > 1$, the survival function $\bar{F}_0(\tau x)$ decays more rapidly in $x$, inducing a steeper change in the location function $\mu(x)$ and consequently amplifying the sensitivity of the conditional survival of $Y$ to changes in $X$. Conversely, when $\tau < 1$, the dependence between $X$ and $Y$ weakens, while $\tau = 0$ implies independence as $\mu(x)$ becomes constant in $x$. \\

Given the bounds established in \eqref{eq: mu'(x) bounds}, and upon selecting the non-decreasing specification for $\mu(x)$ as defined in \eqref{eq: mu+(x) general 2}, and again assuming $\inf_{t: f_1'(t)<0} (-Q(t)) = s^* > 0$, we obtain bounds for $\tau$ as such:
\begin{align}
    0 &\leq \frac{d(\mu^+(x))}{dx} \leq s^* h_0(x)\beta, \nonumber\\
    \therefore 
    0
    &\leq\tau s^*\beta h_0(\tau x)
    \leq s^* h_0(x)\beta, \nonumber\\
    \therefore 
    0 &\leq \tau\frac{h_0(\tau x)}{h_0(x)} \leq 1. \nonumber
\end{align}
We can express this inequality equivalently as: 
\begin{align}
    0 \le r(x, \tau) \le 1,
    \label{eq: tau bounds}
\end{align}
where $r(x, \tau) = \tau \frac{h_0(\tau x)}{h_0(x)}$. The bounds in \eqref{eq: tau bounds} also hold under the non-increasing specification of $\mu(x)$, as defined in \eqref{eq: mu-(x) general 2}. \\
% Theorem \ref{thm: ifr hazard} in Appendix \ref{app: gen theorems} serve as a basis for establishing bounds on the dependence parameter $\tau$ that satisfy \eqref{eq: tau bounds}.\\

\section{Weibull Marginal}
The following sections outline the range of models employed in this study. In particular, we consider a class of mixture models in which the marginal survival function of $X$ follows a Weibull distribution (although any distributional family with positive support could, in principle, be used), while the conditional survival function of $Y$ belongs to a location-scale distributional family. Specifically, we assume that $Y \sim {Family}(\mu(0), \beta)$, where the location parameter is given by $\mu(0) = \gamma \in \mathbb{R}$ and the scale parameter satisfies $\beta > 0$. Given the marginal Weibull specification, the survival function of $X$ is expressed as
\begin{align}
    \bar{F}_0(x) &= P(X>x) = e^{-(\alpha x)^\lambda}, \quad x> 0. \label{eq: weibull marginal}
\end{align}
Consequently, the corresponding density function is $f_X(x) = \alpha \lambda (\alpha x)^{\lambda - 1}e^{-(\alpha x )^\lambda}$ for $x> 0$ and hence, $X\sim Weibull(\alpha, \lambda)$ with:
\begin{align}
    E(X) &= \frac{1}{\alpha} \Gamma\left(1+ \frac{1}{\lambda}\right), \label{eq: Weibull ex} \\
    {Var}(X) &=  \frac{1}{\alpha^2}\left( \Gamma\left( 1+ \frac{2}{\lambda}\right)  - \left(\Gamma\left(1+\frac{1}{\lambda}\right)\right)^2 \right)\label{eq: Weibull varx}.
\end{align}
Furthermore, assuming $\inf_{t: f_1'(t)<0} (-Q(t)) = s^* =1$ and $\sup_{t: f_1'(t)>0} (-Q(t)) = c^* =-1$, we employ \eqref{eq: mu+(x) general 2}, \eqref{eq: mu-(x) general 2} and \eqref{eq: weibull marginal} to obtain:  
\begin{align}
     \mu^+(x) &= \gamma - s^*\beta \log\left( \bar{F}_0 (\tau x)\right) \nonumber \\
     &= \gamma + \beta \left(\alpha \tau x\right)^\lambda, \label{eq: mu+(x) weibull}\\
    \mu^-(x) &= \gamma - \beta \left(\alpha \tau x\right)^\lambda, \label{eq: mu-(x) weibull}
\end{align}
where $\gamma \in \mathbb{R}$ and $\alpha, \beta, \lambda > 0$ with $\tau \in [0, 1]$ for $x>0$. Now with the aid of Theorem \ref{thm: shift identitiy} in Appendix \ref{app: gen theorems}, and expressing the location functions in \eqref{eq: mu+(x) general 2} and \eqref{eq: mu-(x) general 2} as $\mu(x) =  \gamma + \operatorname{sgn}\!\big(\mu'(x)\big)\beta \left(\alpha \tau x\right)^\lambda$ by noting
$\operatorname{sgn}\!\big(\mu'(x)\big)=
    \begin{cases}
        +1, & \text{for } \mu^{+}(x)\ \text{(non-decreasing case)},\\[4pt]
        -1, & \text{for } \mu^{-}(x)\ \text{(non-increasing case)},
    \end{cases}$ we have:
\begin{align}
\operatorname{Cov}(X,Y)
&=\int_{0}^{\infty}
\bar F_0(x)\underbrace{\int_{-\infty}^{\infty}\Big[
\bar F_1\!\Big(\tfrac{y-\mu(x)}{\beta}\Big)-\bar F_1\!\Big(\tfrac{y-\mu(0)}{\beta}\Big)
\Big]\;dy}_{=\ \mu(x)-\mu(0)}\,dx \nonumber\\[4pt]
&=\int_{0}^{\infty}\bar F_0(x)\,[\mu(x)-\mu(0)]\,dx \nonumber\\[4pt]
&=\operatorname{sgn}(\mu'(x))\,\beta\int_{0}^{\infty}  e^{-(\alpha x)^\lambda} \left(\alpha \tau x\right)^\lambda
\,dx
\qquad\left(\text{since }\mu(x)-\mu(0)=\operatorname{sgn}(\mu'(x))\beta \left(\alpha \tau x\right)^\lambda\right)\nonumber\\[4pt]
&=\operatorname{sgn}(\mu'(x))\frac{\beta\tau^\lambda}{\alpha\lambda}\Gamma\left(1+\frac{1}{\lambda}\right),\label{eq: cov general}
\end{align}
where \eqref{eq: cov general} illustrates that higher values of $\beta$ and $\tau$ strengthen the dependence between $X$ and $Y$, whereas larger values of $\alpha$ and $\lambda$ weaken it. Furthermore, as shown by \citet{lakhani2025modelsacceleratedfailureconditionals}, in order for condition~\eqref{eq: tau bounds} to hold, the parameter $\tau$ must satisfy $\tau \in [0, 1]$, given that the Weibull distribution has a hazard function of the form $h_0(x) = \lambda\alpha (\alpha x)^{\lambda - 1} $.\\

Furthermore, the following sections present closed-form expressions of the models' moments, and additionally, theoretical correlation bounds are obtained (with the exception of the Weibull-Cauchy model since moments are undefined).

\subsection{Weibull-Logistic}
We restrict the conditional survival form to the logistic family. Hence for each $x > 0$:
\begin{align}
    \bar{F}_1\left(\frac{y-\mu(x)}{\beta} \right) &= P(Y>y\mid X>x) = 1 - \frac{1}{1+e^{-\frac{y-\mu(x)}{\beta}}} = \frac{1}{1+e^\frac{y-\mu(x)}{\beta}}, \quad x>0, \; y \in \mathbb{R},\nonumber
\end{align}
where $\beta > 0$ with $\mu(x) \in \mathbb{R}$ for all $x>0 $ to ensure a valid survival function. The corresponding joint survival function will be:
\begin{align}
    \bar{F}_0(x) \bar{F}_1\left(\frac{y-\mu(x)}{\beta} \right)  &= P(X>x,Y>y) =  \frac{ e^{-(\alpha x)^\lambda}}{1+e^\frac{y-\mu(x)}{\beta}}, \quad x > 0, \; y \in \mathbb{R}.
    \label{eq: logistic joint survival}
\end{align}
Accordingly, $f_Y(y) =\frac{e^{-\frac{y- \mu(0)}{\beta}}}{\beta\left(1+e^{-\frac{y-\mu(0)}{\beta}}\right)^2}$ for $y\in \mathbb{R}$. Hence, $Y\sim Logistic(\mu(0), \beta)$.\\

Since $f_1'(t) 
= \frac{e^t(1-e^t)}{(e^t+1)^3}
> 0$ for $t<0$ and the function $-Q(t) = \frac{1}{\tanh\left(\frac{t}{2}\right)}$
is decreasing on this interval, it follows that $\sup_{t < 0} (-Q(t)) = \lim_{t \to -\infty} (-Q(t))= c^*= -1$. Additionally, since $f_1'(t) 
= \frac{e^t(1-e^t)}{(e^t+1)^3}
< 0$ for $t>0$ and the function $-Q(t) = \frac{1}{\tanh\left(\frac{t}{2}\right)}$
is decreasing on this interval, it follows that $\inf_{t > 0} (-Q(t)) = \lim_{t \to \infty} (-Q(t)) = s^* =  1$. Hence, applying \eqref{eq: mu'(x) bounds} yields:
\begin{align}
     -h_0(x)\beta \leq \mu'(x) \leq h_0(x)\beta. \nonumber
\end{align}
Furthermore, given that $s^* = 1$ and $c^* = -1$, the location functions \eqref{eq: mu+(x) weibull} and \eqref{eq: mu-(x) weibull} may be employed.
\subsubsection{Moments}
We note that since $Y\sim Logistic(\gamma, \beta)$, we have:
\begin{align}
    E(Y) &= \gamma, \label{eq: ey logistic} \\
    Var(Y) & =\frac{\beta^2\pi^2}{3}. \label{eq: vary logistic}
\end{align}
Now using \eqref{eq: Weibull varx}, \eqref{eq: cov general} and \eqref{eq: vary logistic}, we have:
\begin{align}
    \rho(X, Y) &= \frac{Cov(X, Y)}{\sqrt{Var(X)Var(Y)}} \nonumber \\
    &=  \operatorname{sgn}(\mu'(x))\frac{\sqrt{3}}{\pi}
\frac{\tau^{\lambda}}{\lambda}
\frac{\Gamma\left(1+\frac{1}{\lambda}\right)}
{\sqrt{\Gamma\left(1+\frac{2}{\lambda}\right)
- \Gamma^2\left(1+\frac{1}{\lambda}\right)}}. \label{eq: rho logistic}
\end{align}
Consequently, \eqref{eq: rho logistic} attains its maximum magnitude at $\tau=1$ and its minimum magnitude (zero) at $\tau=0$. Therefore,
\begin{align}
0\le \lvert \rho(X,Y) \rvert \le     \rho_{\text{max}}(\lambda):= \frac{\sqrt{3}}{\pi}
\frac{1}{\lambda}
\frac{\Gamma\left(1+\frac{1}{\lambda}\right)}
{\sqrt{\Gamma\left(1+\frac{2}{\lambda}\right)
- \Gamma^2\left(1+\frac{1}{\lambda}\right)}}.\nonumber
\end{align}
Denoting $\lambda^* = \operatorname*{argmax}_{\lambda>0} \rho_{\max}$, we have $\sup_{\lambda>0}\rho_{\max}(\lambda)
= \rho_{\max}(\lambda^*=1)
= \frac{\sqrt{3}}{\pi}$. Hence, the Weibull-logistic model with location function $\mu^{+}(x)$ in \eqref{eq: mu+(x) weibull} 
can accommodate correlations in the range 
$\rho(X,Y) \in \Big[0,\, \frac{\sqrt{3}}{\pi}\approx 0.5513\Big]$, 
whereas the specification with $\mu^{-}(x)$ in \eqref{eq: mu-(x) weibull} 
can accomodate correlations in the range 
$\rho(X,Y) \in \Big[-\frac{\sqrt{3}}{\pi}\approx -0.5513,\, 0\Big]$.

\subsection{Weibull-Gumbel}
We now assume the conditional survival form belongs to the Gumbel family. Hence for each $x > 0$:
\begin{align}
    \bar{F}_1\left(\frac{y - \mu(x) }{\beta}\right) &=P(Y>y\mid X>x) = 1 - e^{-e^{-\frac{y-\mu(x)}{\beta}}},  \quad x>0,\;  y \in \mathbb{R},\nonumber
\end{align}
where $\beta> 0$ with $\mu(x) \in \mathbb{R}$ for all $x >0 $ to ensure a valid survival function. The corresponding joint survival function will be:
\begin{align}
    \bar{F}_0(x)\bar{F}_1\left(\frac{y - \mu(x) }{\beta}\right) &= P(X>x,Y>y) = e^{-(\alpha x)^\lambda}\left(1 - e^{-e^{-\frac{y-\mu(x)}{\beta}}} \right), \quad x>0,\;  y \in \mathbb{R}. \label{eq: gumbel joint survival}
\end{align}
Accordingly, $f_Y(y) = \frac{e^{-\left(e^{-\frac{y-\mu(0)}{\beta}}+\frac{y-\mu(0)}\beta\right)}}{\beta}$ for $y\in \mathbb{R}$ and hence $Y\sim Logistic(\mu(0), \beta)$.\\

Since $f_1'(t) 
= (1-e^t)e^{-2t-e^{-t}}
> 0$ for $t < 0$ and the function $-Q(t) = -\frac{e^t}{1-e^t}$
is decreasing on this interval, $\sup_{t<0}(-Q(t)) = \lim_{t\to -\infty}(-Q(t)) = c^* = 0$. Additionally, since $f_1'(t) 
= (1-e^t)e^{-2t-e^{-t}}
< 0$ for $t > 0$ and $-Q(t) = -\frac{e^t}{1-e^t}$
is decreasing on this interval, $\inf_{t>0}(-Q(t)) = \lim_{t\to \infty}(-Q(t)) = s^* = 1$. Hence, applying \eqref{eq: mu'(x) bounds} yields:
\begin{align}
    0\leq \mu'(x) \leq h_0(x)\beta. \nonumber
\end{align}
Hence, given that $s^* = 1$, the location function \eqref{eq: mu+(x) weibull} may be utilised. Moreover, since $c^* = 0$, the model is not capable of representing negative correlations.

\subsubsection{Moments}
We note that since $Y\sim Gumbel(\gamma, \beta)$, we have:
\begin{align}
    E(Y) &= \gamma +\beta\gamma^*, \label{eq: ey gumbel} \\
    Var(Y) & = \frac{\pi^2}{6}\beta^2, \label{eq: vary gumbel}
\end{align}
where $\gamma^*$ denotes the Euler-Mascheroni constant. Now using \eqref{eq: Weibull varx}, \eqref{eq: cov general} and \eqref{eq: vary gumbel}, we have:
\begin{align}
    \rho(X, Y) &= \frac{Cov(X, Y)}{\sqrt{Var(X)Var(Y)}} \nonumber \\
    &=\frac{\sqrt{6}}{\pi}
\frac{\tau^{\lambda}}{\lambda}
\frac{\Gamma\left(1+\frac{1}{\lambda}\right)}
{\sqrt{\Gamma\left(1+\frac{2}{\lambda}\right)
- \Gamma^2\left(1+\frac{1}{\lambda}\right)}}. \label{eq: rho gumbel}
\end{align}
Consequently, \eqref{eq: rho gumbel} attains its maximum magnitude at $\tau=1$ and its minimum magnitude (zero) at $\tau=0$. Therefore,
\begin{align}
0\le \lvert \rho(X,Y) \rvert \le     \rho_{\text{max}}(\lambda):= \frac{\sqrt{6}}{\pi}
\frac{1}{\lambda}
\frac{\Gamma\left(1+\frac{1}{\lambda}\right)}
{\sqrt{\Gamma\left(1+\frac{2}{\lambda}\right)
- \Gamma^2\left(1+\frac{1}{\lambda}\right)}}.\nonumber
\end{align}
Denoting $\lambda^* = \operatorname*{argmax}_{\lambda>0} \rho_{\max}$, we have $\sup_{\lambda>0}\rho_{\max}(\lambda)
= \rho_{\max}(\lambda^*=1)
= \frac{\sqrt{6}}{\pi}$. Hence, the Weibull-Gumbel model with location function $\mu^{+}(x)$ in \eqref{eq: mu+(x) weibull} 
can accommodate correlations in the range 
$\rho(X,Y) \in \Big[0,\, \tfrac{\sqrt{6}}{\pi} \approx 0.7797\Big]$.

\subsection{Weibull-Laplace}
We now restrict the conditional survival form to the Laplace family. Hence for each $x > 0$:
\begin{align}
    \bar{F}_1\left(\frac{y - \mu(x) }{\beta}\right) &= P(Y>y\mid X>x)  = \frac{1}{2} - \frac{1}{2}\operatorname{sgn}(y- \mu(x))\left(1 - e^{-\frac{\lvert y-\mu(x) \rvert}{\beta}} \right),  \quad x>0,\;  y \in \mathbb{R},\nonumber
\end{align}
where $\beta> 0$ with $\mu(x) \in \mathbb{R}$ for all $x >0$ to ensure a valid survival function. The corresponding joint survival function will be:
\begin{align}
    \bar{F}_0(x)\bar{F}_1\left(\frac{y - \mu(x) }{\beta}\right)=
    &P(X>x,Y>y) = e^{-(\alpha x)^\lambda} \left( \frac{1}{2} - \frac{1}{2}\operatorname{sgn}(y- \mu(x))\left(1 - e^{-\frac{\lvert y-\mu(x) \rvert}{\beta}} \right) \right), \quad x > 0,\; y \in \mathbb{R}. \label{eq: laplace joint survival}
\end{align}
Accordingly, $f_Y(y) = \frac{1}{2\beta}e^{-\frac{\lvert y-\mu(0)\rvert}{\beta}}$ for $y\in \mathbb{R}$ hence $Y\sim Laplace(\mu(0), \beta)$.\\

Since $f_1'(t) 
= -\frac{1}{2}\operatorname{sgn}(t)e^{-\lvert t\rvert}
> 0$ for $t < 0$ and $-Q(t) = \operatorname{sgn}(t)$, $\sup_{t<0}(-Q(t)) = c^*= -1$. Additionally, since $f_1'(t) 
= -\frac{1}{2}\operatorname{sgn}(t)e^{-\lvert t\rvert}
< 0$ for $t > 0$ and $-Q(t) =\operatorname{sgn}(t)$, $\inf_{t>0}(-Q(t)) = s^*= 1$. Hence, applying \eqref{eq: mu'(x) bounds} yields:
\begin{align}
    -h_0(x)\beta \leq \mu'(x) \leq h_0(x)\beta . \nonumber
\end{align}
Hence, given that $s^* = 1$ and $c^* = -1$, the location functions \eqref{eq: mu+(x) weibull} and \eqref{eq: mu-(x) weibull} may be employed.

\subsubsection{Moments}
We note that since $Y\sim Laplace(\gamma, \beta)$, we have:
\begin{align}
    E(Y) &= \gamma, \label{eq: ey laplace} \\
    Var(Y) & = 2\beta^2. \label{eq: vary laplace}
\end{align}
Now using \eqref{eq: Weibull varx}, \eqref{eq: cov general} and \eqref{eq: vary laplace}, we have:
\begin{align}
    \rho(X, Y) &= \frac{Cov(X, Y)}{\sqrt{Var(X)Var(Y)}} \nonumber \\
    &= \operatorname{sgn}(\mu'(x))\frac{1}{\sqrt{2}}
\frac{\tau^{\lambda}}{\lambda}
\frac{\Gamma\left(1+\frac{1}{\lambda}\right)}
{\sqrt{\Gamma\left(1+\frac{2}{\lambda}\right)
- \Gamma^2\left(1+\frac{1}{\lambda}\right)}}. \label{eq: rho laplace}
\end{align}
Consequently, \eqref{eq: rho laplace} attains its maximum magnitude at $\tau=1$ and its minimum magnitude (zero) at $\tau=0$. Therefore,
\begin{align}
0\le \lvert \rho(X,Y) \rvert \le      \rho_{\text{max}}(\lambda):= \frac{1}{\sqrt{2}}
\frac{1}{\lambda}
\frac{\Gamma\left(1+\frac{1}{\lambda}\right)}
{\sqrt{\Gamma\left(1+\frac{2}{\lambda}\right)
- \Gamma^2\left(1+\frac{1}{\lambda}\right)}}.\nonumber
\end{align}
Denoting $\lambda^* = \operatorname*{argmax}_{\lambda>0} \rho_{\max}$, we have $\sup_{\lambda>0}\rho_{\max}(\lambda)
= \rho_{\max}(\lambda^*=1)
= \frac{1}{\sqrt{2}}$. Hence, the Weibull-Laplace model with location function $\mu^{+}(x)$ in \eqref{eq: mu+(x) weibull} 
can accommodate correlations in the range 
$\rho(X,Y)\in\left[0, \frac{1}{\sqrt{2}}\approx 0.7071 \right]$, 
whereas the specification with $\mu^{-}(x)$ in \eqref{eq: mu-(x) weibull} 
can accomodate correlations in the range 
$\rho(X,Y)\in\left[-\frac{1}{\sqrt{2}}\approx -0.7071,0 \right]$.

\subsection{Weibull-Cauchy}
We now assume the conditional survival form belongs to the Cauchy family. Hence for each $x > 0$:
\begin{align}
    \bar{F}_1\left(\frac{y - \mu(x) }{\beta}\right) &=P(Y>y\mid X>x) =  \frac{1}{2}-\frac{1}{\pi}\arctan\left(\frac{y-\mu(x)}{\beta}\right),  \quad x>0, \; y \in \mathbb{R},\nonumber
\end{align}
where $\beta> 0$ with $\mu(x) \in \mathbb{R}$ for all $x>0$ to ensure a valid survival function. The corresponding joint survival function will be:
\begin{align}
    \bar{F}_0(x)\bar{F}_1\left(\frac{y - \mu(x) }{\beta}\right) &= P(X>x,Y>y) =  e^{-(\alpha x)^\lambda}\left(  \frac{1}{2}-\frac{1}{\pi}\arctan\left(\frac{y-\mu(x)}{\beta}\right)\right), \quad x>0, \; y \in \mathbb{R}. 
    \label{eq: cauchy joint survival}
\end{align}
Accordingly, $f_Y(y) = \frac{1}{\pi\beta\left(1+\left(\frac{y-\mu(0)}{\beta}\right)^2 \right)}$ for $y\in \mathbb{R}$. Hence,  $Y\sim Cauchy(\mu(0), \beta)$.\\

Since $f_1'(t) 
= -\frac{2t}{\pi(t^2+1)^2}
> 0$ for $t < 0$
it follows that the function $-Q(t) = \frac{t^2+1}{2t}$ is increasing on $(-\infty, -1)$ and decreasing on $(-1, 0)$. Hence, $-Q(t)$ attains its maximum at $t = -1$, where $-Q(-1) = -1$. Consequently, $\sup_{t<0}(-Q(t)) = c^* = -1$. Additionally, since $f_1'(t) 
= -\frac{2t}{\pi(t^2+1)^2}
< 0$ for $t > 0$
it follows that the function $-Q(t) = \frac{t^2+1}{2t}$ is decreasing on $(0, 1)$ and increasing on $(1, \infty)$. Hence, $-Q(t)$ attains its minimum at $t = 1$, where $-Q(1) = 1$. Hence, $\inf_{t>0}(-Q(t)) = s^*= 1$. Now, applying \eqref{eq: mu'(x) bounds} yields:
\begin{align}
    -h_0(x)\beta \leq \mu'(x) \leq h_0(x)\beta . \nonumber
\end{align}
Consequently, given that $s^* = 1$ and $c^* = -1$, the location functions \eqref{eq: mu+(x) weibull} and \eqref{eq: mu-(x) weibull} may be employed.

\subsection{Weibull-Normal}
Lastly, we restrict the conditional survival form to the normal family. Hence for each $x > 0$:
\begin{align}
    \bar{F}_1\left(\frac{y-\mu(x)}{\beta} \right) &= P(Y>y\mid X>x) =  \frac{1}{2}\left(1 - \operatorname{erf}\left( \frac{y - \mu(x)}{\beta\sqrt{2}} \right) \right), \quad x>0, \; y \in \mathbb{R}, \nonumber
\end{align}
where $\beta > 0$ and $\mu(x) \in \mathbb{R}$ for all $x>0$ to ensure a valid survival function. The corresponding joint survival function will be:
\begin{align}
    \bar{F}_0(x)   \bar{F}_1\left(\frac{y-\mu(x)}{\beta} \right)&= P(X>x,Y>y) = \frac{e^{-(\alpha x)^\lambda}}{2}\left(1 - \operatorname{erf}\left( \frac{y - \mu(x)}{\beta\sqrt{2}} \right) \right), \quad x>0, \; y \in \mathbb{R}. \nonumber
    % \label{eq: normal joint survival}
\end{align}
Accordingly, $f_Y(y) = \frac{e^{-\frac{(y-\mu(0))^2}{2\beta^2}}}{\sqrt{2\pi\beta^2}}$ for $y \in \mathbb{R}$. Hence $Y\sim Normal(\mu(0), \beta)$.\\

Since $f_1'(t) 
= -\frac{te^{-\frac{t^2}{2}}}{\sqrt{2\pi}} 
> 0$ for $t <0 $ and the function $-Q(t) = \frac{1}{t}$ is decreasing on this interval, it follows that $\sup_{t <0}(-Q(t)) = \lim_{t \to -\infty} (-Q(t)) =c^* =  0$. Additionally, since $f_1'(t) 
= -\frac{te^{-\frac{t^2}{2}}}{\sqrt{2\pi}} 
< 0$ for $t >0 $ and the function $-Q(t) = \frac{1}{t}$ is decreasing on this interval, it follows that $\inf_{t >0 }(-Q(t)) = \lim_{t \to \infty} (-Q(t)) = s^*= 0$. Consequently, applying \eqref{eq: mu'(x) bounds} yields:
\begin{align}
    0  \leq \mu'(x) \leq 0. \nonumber
\end{align}
Thus, the location function reduces to a constant form, $\mu(x) = \gamma$ for some $\gamma \in \mathbb{R}$, implying statistical independence between $X$ and $Y$. As a result, employing the normal model with this location function is futile for representing dependence between the variables. Consequently, the study forgoes further utilisation of the normal model.

\section{Simulation}
We utilise the Metropolis-Hastings (MH) algorithm\footnote{Presented in Appendix \ref{app: MH}} to obtain $X$ and $Y$ (where $Y$ is conditioned on $\{X>x\}$) draws by sampling from the joint density of the models considered in this study. The target distribution is given by $\pi(x,y)=f_{X,Y}(x,y)$, noting that $X\sim Weibull(\alpha, \lambda)$ and $Y \sim Family(\mu(0), \beta)$, with $f_{X,Y}(x,y)$ given in Appendix \ref{app: joint lik}. Additionally, Appendix \ref{app: joint lik} presents the log-likelihoods for the study's models - m.l.e's are obtained numerically as $\operatorname{argmax}_{\boldsymbol{\theta}} \ell (\boldsymbol{\theta})$ where $\boldsymbol{\theta} = \left[\alpha, \beta, \lambda, \gamma, \tau\right]'$. We employ $\mu^{+}(x)$ as defined in \eqref{eq: mu+(x) weibull} for this section, although $\mu^{-}(x)$, given in \eqref{eq: mu-(x) weibull}, could equivalently have been used.\\

Furthermore, we conduct simulations based on $10000$ datasets of size $n = 100, 500$ and $1000$ for fixed parameter values: $\alpha = 1, \beta =2, \lambda =3, \gamma = -4$ and $\tau = 0.5$.

\subsection{Simulation Data}
The corresponding m.m.e and m.l.e results are reported in Tables~\ref{Table: Simulations for logistic model} - \ref{Table: Simulations for cauchy model}, together with $95\%$ confidence intervals computed as $\hat{\theta} \pm Z_{\alpha/2}\text{S.E.}(\hat{\theta})$, where $\hat{\theta}$ denotes the point estimator \footnote{The reported standard errors (S.E.) are calculated as the empirical standard deviation of the parameter estimates across simulated datasets, representing the sampling variability of the estimator.}. In addition, we report Pearson correlation coefficients denoted as $\text{PC}$.\\

As expected, the standard errors of both m.m.e's and m.l.e's \footnote{m.m.e's (where applicable) were used as initial values when using R's \texttt{optim} function. Furthermore, Appendix \ref{app: mmes} provides a guide on how to obtain m.m.e's for the models (noting that moments were not provided for the Weibull-Cauchy model, we exclude m.m.e's for said model).} decrease as the sample size $n$ increases, yielding narrower $95\%$ confidence intervals. The simulation results indicate that no instability arises in the parameter estimation; however, as expected, more reliable estimation requires larger sample sizes. Furthermore, the population correlations converge to the PCs as sample size increases.\\

Figure \ref{fig: sims gif} displays the bivariate densities of the models used in the study, given the fixed parameter set.

\begin{table}[H] \centering 
	\begin{tabular}{@{\extracolsep{1pt}} ccccccccc} 
		\\[-1.8ex]\hline 
		\hline \\[-1.8ex] 
		n & Parameter & m.m.e & m.l.e & SE(m.m.e) & SE(m.l.e) & $95\%$CI (m.m.e) & $95\%$CI (m.l.e) & PC \\ 
		\hline \\[-1.8ex] 
	100 & $\alpha$ & 0.999 & 0.999 & 0.038 & 0.038 & $(0.925,\,1.073)$ & $(0.925,\,1.073)$ & 0.101 \\ 
	     & $\beta$ & 1.977 & 1.989 & 0.186 & 0.177 & $(1.612,\,2.342)$ & $(1.642,\,2.335)$ &  \\ 
	     & $\lambda$ & 3.024 & 3.031 & 0.257 & 0.254 & $(2.520,\,3.527)$ & $(2.533,\,3.529)$ &  \\ 
	     & $\gamma$ & -3.990 & -3.991 & 0.389 & 0.372 & $(-4.752,\,-3.227)$ & $(-4.721,\,-3.261)$ &  \\ 
	     & $\tau$ & 0.587 & 0.577 & 0.157 & 0.189 & $(0.280,\,0.894)$ & $(0.206,\,0.948)$ &  \\ 
	     & $\rho$ & 0.101 & 0.096 & 0.078 & 0.077 & $(-0.052,\,0.254)$ & $(-0.056,\,0.247)$ &  \\[1ex]
\hline 
	500 & $\alpha$ & 1.000 & 1.000 & 0.017 & 0.017 & $(0.967,\,1.033)$ & $(0.967,\,1.033)$ & 0.068 \\ 
	     & $\beta$ & 1.997 & 2.000 & 0.084 & 0.079 & $(1.832,\,2.161)$ & $(1.845,\,2.155)$ &  \\ 
	     & $\lambda$ & 3.004 & 3.006 & 0.113 & 0.112 & $(2.783,\,3.225)$ & $(2.787,\,3.225)$ &  \\ 
	     & $\gamma$ & -3.999 & -3.999 & 0.173 & 0.165 & $(-4.338,\,-3.661)$ & $(-4.321,\,-3.676)$ &  \\ 
	     & $\tau$ & 0.514 & 0.512 & 0.115 & 0.122 & $(0.287,\,0.740)$ & $(0.272,\,0.752)$ &  \\ 
	     & $\rho$ & 0.068 & 0.067 & 0.041 & 0.039 & $(-0.013,\,0.149)$ & $(-0.010,\,0.145)$ &  \\[1ex]
\hline
	1000 & $\alpha$ & 1.000 & 1.000 & 0.012 & 0.012 & $(0.977,\,1.023)$ & $(0.977,\,1.024)$ & 0.064  \\ 
	     & $\beta$ & 1.998 & 2.000 & 0.060 & 0.056 & $(1.881,\,2.115)$ & $(1.891,\,2.109)$ &  \\ 
	     & $\lambda$ & 3.003 & 3.004 & 0.080 & 0.080 & $(2.846,\,3.161)$ & $(2.847,\,3.161)$ &  \\ 
	     & $\gamma$ & -3.997 & -3.998 & 0.122 & 0.116 & $(-4.236,\,-3.758)$ & $(-4.225,\,-3.771)$ &  \\ 
	     & $\tau$ & 0.503 & 0.503 & 0.096 & 0.094 & $(0.316,\,0.691)$ & $(0.318,\,0.688)$ &  \\ 
	     & $\rho$ & 0.064 & 0.064 & 0.032 & 0.030 & $(0.001,\,0.128)$ & $(0.004,\,0.124)$ &  \\[1ex]
		\hline
	\end{tabular} 
    \caption{Simulations for Weibull-Logistic model ($\alpha = 1$, $\beta =2$, $\lambda =3$, $\gamma = -4$ and $\tau = 0.5$ with $\rho = 0.063$).} 
    \label{Table: Simulations for logistic model} 
\end{table} 

\begin{table}[H] \centering 
	\begin{tabular}{@{\extracolsep{5pt}} ccccccccc} 
		\\[-1.8ex]\hline 
		\hline \\[-1.8ex] 
		n & Parameter & m.m.e & m.l.e & SE(m.m.e) & SE(m.l.e) & $95\%$CI (m.m.e) & $95\%$CI (m.l.e) & PC \\ 
		\hline \\[-1.8ex] 
	100 & $\alpha$ & 0.999 & 1.000 & 0.039 & 0.039 & $(0.923,\,1.077)$ & $(0.922,\,1.077)$ & 0.116  \\ 
	     & $\beta$ & 1.964 & 1.983 & 0.219 & 0.174 & $(1.536,\,2.393)$ & $(1.642,\,2.325)$ &  \\ 
	     & $\lambda$ & 3.029 & 3.034 & 0.265 & 0.263 & $(2.510,\,3.548)$ & $(2.519,\,3.550)$ &  \\ 
	     & $\gamma$ & -3.987 & -3.996 & 0.239 & 0.235 & $(-4.455,\,-3.519)$ & $(-4.457,\,-3.535)$ &  \\ 
	     & $\tau$ & 0.549 & 0.541 & 0.141 & 0.144 & $(0.272,\,0.826)$ & $(0.260,\,0.823)$ &  \\ 
	     & $\rho$ & 0.116 & 0.112 & 0.085 & 0.072 & $(-0.051,\,0.283)$ & $(-0.029,\,0.254)$ &  \\[1ex]
    \hline
	500 & $\alpha$ & 1.000 & 1.000 & 0.018 & 0.018 & $(0.966,\,1.035)$ & $(0.966,\,1.035)$ & 0.090  \\ 
	     & $\beta$ & 1.994 & 1.997 & 0.101 & 0.077 & $(1.796,\,2.192)$ & $(1.847,\,2.148)$ &  \\ 
	     & $\lambda$ & 3.007 & 3.009 & 0.116 & 0.116 & $(2.779,\,3.235)$ & $(2.782,\,3.237)$ &  \\ 
	     & $\gamma$ & -3.997 & -3.999 & 0.106 & 0.105 & $(-4.205,\,-3.789)$ & $(-4.203,\,-3.794)$ &  \\ 
	     & $\tau$ & 0.503 & 0.505 & 0.098 & 0.070 & $(0.311,\,0.695)$ & $(0.368,\,0.642)$ &  \\ 
	     & $\rho$ & 0.090 & 0.091 & 0.047 & 0.035 & $(-0.002,\,0.182)$ & $(0.022,\,0.160)$ &  \\[1ex]
\hline 
	1000 & $\alpha$ & 1.000 & 1.000 & 0.012 & 0.012 & $(0.976,\,1.025)$ & $(0.976,\,1.025)$ & 0.089 \\ 
	     & $\beta$ & 1.996 & 1.998 & 0.072 & 0.055 & $(1.856,\,2.137)$ & $(1.891,\,2.105)$ &  \\ 
	     & $\lambda$ & 3.004 & 3.004 & 0.083 & 0.083 & $(2.841,\,3.166)$ & $(2.841,\,3.167)$ &  \\ 
	     & $\gamma$ & -3.998 & -3.999 & 0.075 & 0.074 & $(-4.145,\,-3.852)$ & $(-4.143,\,-3.854)$ &  \\ 
	     & $\tau$ & 0.500 & 0.501 & 0.075 & 0.049 & $(0.352,\,0.647)$ & $(0.405,\,0.598)$ &  \\ 
	     & $\rho$ & 0.089 & 0.090 & 0.036 & 0.026 & $(0.019,\,0.159)$ & $(0.039,\,0.140)$ &  \\[1ex]
\hline
	\end{tabular} 
    \caption{Simulations for Weibull-Gumbel model ($\alpha = 1$, $\beta =2$, $\lambda =3$, $\gamma = -4$ and $\tau = 0.5$ with $\rho = 0.089$).} 
    \label{Table: Simulations for gumbel model} 
\end{table} 

\begin{table}[H] \centering 
	\begin{tabular}{@{\extracolsep{5pt}} ccccccccc} 
		\\[-1.8ex]\hline 
		\hline \\[-1.8ex] 
		n & Parameter & m.m.e & m.l.e & SE(m.m.e) & SE(m.l.e) & $95\%$CI (m.m.e) & $95\%$CI (m.l.e) & PC \\ 
		\hline \\[-1.8ex] 
	100 & $\alpha$ & 0.999 & 0.997 & 0.040 & 0.042 & $(0.921,\,1.078)$ & $(0.916,\,1.079)$ & 0.112 \\ 
	     & $\beta$ & 1.963 & 1.971 & 0.235 & 0.222 & $(1.502,\,2.423)$ & $(1.535,\,2.406)$ &  \\ 
	     & $\lambda$ & 3.029 & 3.051 & 0.269 & 0.274 & $(2.502,\,3.556)$ & $(2.513,\,3.588)$ &  \\ 
	     & $\gamma$ & -3.988 & -4.062 & 0.307 & 0.259 & $(-4.591,\,-3.385)$ & $(-4.570,\,-3.555)$ &  \\ 
	     & $\tau$ & 0.561 & 0.526 & 0.144 & 0.159 & $(0.279,\,0.844)$ & $(0.214,\,0.839)$ &  \\ 
	     & $\rho$ & 0.112 & 0.092 & 0.083 & 0.074 & $(-0.050,\,0.274)$ & $(-0.053,\,0.236)$ &  \\[1ex]
    \hline
	500 & $\alpha$ & 1.000 & 0.999 & 0.018 & 0.019 & $(0.965,\,1.036)$ & $(0.962,\,1.036)$ &  0.083\\ 
	     & $\beta$ & 1.994 & 1.996 & 0.107 & 0.101 & $(1.785,\,2.203)$ & $(1.798,\,2.194)$ &  \\ 
	     & $\lambda$ & 3.008 & 3.016 & 0.119 & 0.122 & $(2.776,\,3.241)$ & $(2.778,\,3.254)$ &  \\ 
	     & $\gamma$ & -3.997 & -4.022 & 0.137 & 0.119 & $(-4.267,\,-3.728)$ & $(-4.256,\,-3.788)$ &  \\ 
	     & $\tau$ & 0.506 & 0.492 & 0.102 & 0.093 & $(0.306,\,0.706)$ & $(0.310,\,0.675)$ &  \\ 
	     & $\rho$ & 0.083 & 0.077 & 0.045 & 0.040 & $(-0.005,\,0.171)$ & $(-0.001,\,0.155)$ &  \\[1ex]
    \hline
	1000 & $\alpha$ & 1.000 & 1.000 & 0.012 & 0.012 & $(0.976,\,1.025)$ & $(0.976,\,1.025)$ & 0.089 \\ 
	     & $\beta$ & 1.996 & 1.998 & 0.072 & 0.055 & $(1.856,\,2.137)$ & $(1.891,\,2.105)$ &  \\ 
	     & $\lambda$ & 3.004 & 3.004 & 0.083 & 0.083 & $(2.841,\,3.166)$ & $(2.841,\,3.167)$ &  \\ 
	     & $\gamma$ & -3.998 & -3.999 & 0.075 & 0.074 & $(-4.145,\,-3.852)$ & $(-4.143,\,-3.854)$ &  \\ 
	     & $\tau$ & 0.500 & 0.501 & 0.075 & 0.049 & $(0.352,\,0.647)$ & $(0.405,\,0.598)$ &  \\ 
	     & $\rho$ & 0.089 & 0.090 & 0.036 & 0.026 & $(0.019,\,0.159)$ & $(0.039,\,0.140)$ &  \\[1ex]
         \hline
	\end{tabular} 
    \caption{Simulations for Weibull-Laplace model ($\alpha = 1$, $\beta =2$, $\lambda =3$, $\gamma = -4$ and $\tau = 0.5$ with $\rho =  0.081$).} 
    \label{Table: Simulations for laplace model} 
\end{table} 

\begin{table}[H] \centering 
	\begin{tabular}{@{\extracolsep{5pt}} cccccc} 
		\\[-1.8ex]\hline 
		\hline \\[-1.8ex] 
		n & Parameter & m.l.e  & SE(m.l.e) & $95\%$CI (m.l.e) & PC \\ 
		\hline \\[-1.8ex] 
	100 & $\alpha$ & 1.000 & 0.037 & $(0.927,\,1.073)$ &  0.076 \\ 
	     & $\beta$ & 1.969 & 0.288 & $(1.404,\,2.535)$ &  \\ 
	     & $\lambda$ & 3.034 & 0.247 & $(2.550,\,3.517)$ &  \\ 
	     & $\gamma$ & -4.008 & 0.296 & $(-4.587,\,-3.428)$ &  \\ 
	     & $\tau$ & 0.504 & 0.256 & $(0.002,\,1.005)$ &  \\[1ex]
    \hline
	500 & $\alpha$ & 1.000 & 0.017 & $(0.967,\,1.033)$ & 0.033 \\ 
	     & $\beta$ & 1.991 & 0.128 & $(1.741,\,2.242)$ &  \\ 
	     & $\lambda$ & 3.008 & 0.112 & $(2.788,\,3.227)$ &  \\ 
	     & $\gamma$ & -4.002 & 0.131 & $(-4.259,\,-3.745)$ &  \\ 
	     & $\tau$ & 0.482 & 0.208 & $(0.076,\,0.889)$ &  \\[1ex]
\hline
	1000 & $\alpha$ & 1.000 & 0.012 & $(0.977,\,1.024)$ & 0.024 \\ 
	     & $\beta$ & 1.996 & 0.090 & $(1.820,\,2.172)$ &  \\ 
	     & $\lambda$ & 3.004 & 0.080 & $(2.847,\,3.160)$ &  \\ 
	     & $\gamma$ & -4.002 & 0.091 & $(-4.180,\,-3.823)$ &  \\ 
	     & $\tau$ & 0.479 & 0.216 & $(0.056,\,0.901)$ &  \\[1ex]
\hline
	\end{tabular} 
    \caption{Simulations for Weibull-Cauchy model ($\alpha = 1$, $\beta =2$, $\lambda =3$, $\gamma = -4$ and $\tau = 0.5$).} 
    \label{Table: Simulations for cauchy model} 
\end{table} 

\begin{figure}[H]
  \centering
  \begin{minipage}[t]{0.48\textwidth}
    % \vspace{0.1em}
    \centering
\animategraphics[controls,autoplay,loop,width=\textwidth]{1}{Sims/pics/}{1}{4}
    % \caption{Bivariate densities of models.}
  \end{minipage}
  \hfill
  \begin{minipage}[t]{0.48\textwidth}
\centering    \animategraphics[controls,autoplay,loop,width=\textwidth]{1}{Sims/2d_pics/}{1}{4}
    % \subcaption{Contour plots of bivariate densities of models (a single dataset of $n = 500$ used; data in black circles).}
  \end{minipage}

  \caption{Bivariate densities of the models with parameter values: $\alpha = 1$, $\beta = 2$, $\lambda = 3$, $\gamma = -4$, and $\tau = 0.5$. The left panel displays the three-dimensional bivariate density surfaces, while the right panel presents the corresponding contour plots (an arbitrary dataset of $n = 500$ used; data in black circles).}
  \label{fig: sims gif}
\end{figure}

\section{A Small Application}
The dataset employs the logarithm of the Volatility Index (VIX) as the independent variable $X$, and the logarithm of daily returns of Nasdaq-100 E-mini futures (NQ) as the dependent variable $Y$ conditioned on $\{X>x\}$, covering the period from 1 January 2020 to 17 October 2025 \footnote{Observations containing missing values were removed. These omissions primarily arose from non-trading days (weekends and public holidays) and discrepancies in trading calendars between the VIX and NQ futures contracts.}. The data was obtained from \citet{yahoo2025finance}. The Pearson correlation between $X$ and $Y$ is $-0.140$; consequently, the Weibull-Gumbel model is unsuitable for this dataset, as previously noted, such a specification cannot accommodate negative correlations.\\

The corresponding m.m.e's, m.l.e's, and AIC values for the VIX/NQ dataset is reported in Table \ref{table: vix}. It is evident from the presented AIC values that the Weibull-Laplace model provides the best fit. This outcome is not entirely unexpected: as can be observed from the histogram of the logarithm of daily Nasdaq-100 returns, in Figure~\ref{fig: application hists gif}, the empirical distribution visually resembles that of a Laplace distribution. This finding underscores the importance of conducting formal goodness-of-fit tests for the marginal variables $X$ and $Y$ prior to estimating the joint models proposed in this study. In particular, given that the modeling framework assumes $X \sim {Weibull}(\alpha, \lambda)$ and $Y \sim {Family}(\mu(0), \beta)$, it is prudent to verify that the empirical distributions of $X$ and $Y$ from the data are consistent with the respective assumed distributional families.\\

Furthermore, the Weibull parameters $\alpha$ and $\lambda$ reported in Table~\ref{table: vix} are estimated almost identically across the three models. This indicates that all three models produce nearly identical Weibull fits for the logarithm of the VIX, as illustrated in Figure~\ref{fig: application hists gif}. In retrospect, it may have been more appropriate to construct the study’s models using a gamma marginal for $X$, which could perhaps offer a better fit for the logarithm of the VIX.\\

Estimated bivariate densities for each model, along with their corresponding contour plots, using m.l.e's from Table \ref{table: vix}, are presented in Figure \ref{fig: application1 gif}.

\begin{table}[H]
    \centering
    \begin{tabular}{ccccc}
\hline
     Model & Parameter & m.m.e & m.l.e  & AIC \\
     \hline
    Weibull-Logistic  &  $\alpha$ & 0.319 & 0.318 & -6958.266 \\
    & $\beta$ & 0.009 & 0.008 & \\
    & $\lambda$ & 11.572 & 8.694 \\
    & $\gamma$ & 0.001 & 0.001 & \\
    & $\tau$ & 0.903 & 0.848  & \\
    & $\rho$ & -0.140 & -0.111 & \\[1ex]
\hline 
    Weibull-Laplace  & $\alpha$ & 0.319 & 0.318 & $\mathbf{-6997.162}$ \\
    & $\beta$ & 0.011 & 0.011 & \\
    & $\lambda$ & 11.572 & 8.688 \\
    & $\gamma$ & 0.001 & 0.001 & \\
    & $\tau$ & 0.884 & 0.749  & \\
    & $\rho$ & -0.140 & -0.048 & \\[1ex]
\hline 
    Weibull-Cauchy & $\alpha$ & - & 0.318 & -6737.427 \\
    & $\beta$ & - & 0.008 & \\
    & $\lambda$ & - & 8.732 \\
    & $\gamma$ & - & 0.001 & \\
    & $\tau$ & - & 0.800  & \\
\hline
    \end{tabular}
    \caption{m.m.e's, m.l.e's and AIC's of models for the VIX/NQ dataset.}
    \label{table: vix}
\end{table}

\begin{figure}[H]
  \centering
\animategraphics[controls,autoplay,loop,width=0.85\textwidth]{1}{Application/hists/}{1}{3}
 \caption{Histogram of the logarithm of the VIX (left) with estimated marginal Weibull densities based on m.l.e's reported in Table~\ref{table: vix}, and histogram of the logarithm of daily Nasdaq-100 (NQ) returns (right) with estimated marginal logistic, Laplace, and Cauchy densities, also based on the corresponding m.l.e's in Table~\ref{table: vix}.}
  \label{fig: application hists gif}
\end{figure}

\begin{figure}[H]
  \centering
  \begin{minipage}[t]{0.48\textwidth}
    \centering
\animategraphics[controls,autoplay,loop,width=\textwidth]{1}{Application/pics/}{1}{3}
    % \caption{Bivariate densities of models.}
  \end{minipage}
  \hfill
  \begin{minipage}[t]{0.48\textwidth}
\centering    \animategraphics[controls,autoplay,loop,width=\textwidth]{1}{Application/2d_pics/}{1}{3}
    % \subcaption{Contour plots of bivariate densities of models (a single dataset of $n = 500$ used; data in black circles).}
  \end{minipage}
  \caption{Estimated bivariate densities of the models using m.l.e's from Table \ref{table: vix}. The left panel displays the estimated 3-D bivariate density surfaces, while the right panel presents the corresponding contour plots (observations in black circles).}
  \label{fig: application1 gif}
\end{figure}

\appendix
\section*{Appendices}
\section{General Theorems} \label{app: gen theorems}
\begin{theorem} \label{thm: cov(x, y) symmetric}
Consider an accelerated conditionals model of the form: 
\[
P(X>x) = \bar{F}_0(x)
\quad \text{and} \quad
P(Y>y \mid X>x) = \bar{F}_1\left(\frac{y-\mu}{\beta(x)}\right),
\]
where $\bar{F}_0$ and $\bar{F}_1$ are survival functions. If the marginal of $Y$ is symmetric about $\mu$, then ${Cov}(X,Y) = 0$.
\end{theorem}
\begin{proof}
Since $P(Y>y , X> x) = \bar{F}_0(x)\bar{F}_1\left(\frac{y-\mu}{\beta(x)} \right)$ then $P(Y>y) = \lim_{x\to 0^+}P(Y>y , X> x) = \bar{F}_0(0)\bar{F}_1\left(\frac{y-\mu}{\beta(0)} \right) = \bar{F}_1\left(\frac{y-\mu}{\beta(0)}\right)$ since $\lim_{x\to 0^+}\bar{F}_0(x) = \bar{F}_0(0)  =  1$ and $\beta(0): = \lim_{x\to 0^+} \beta(x)$. We have from Hoeffding's covariance identity:
\begin{align}
    {Cov}(X, Y)
    &= \int_{0}^\infty \int_{-\infty}^\infty
    \left( P(Y>y, X>x) - P(X>x)P(Y>y) \right)\, dy\, dx \nonumber \\
    &= \int_{0}^\infty \int_{-\infty}^\infty
    \bar{F}_0(x)
    \left\{
        \bar{F}_1\!\left(\frac{y-\mu}{\beta(x)}\right)
        - \bar{F}_1\!\left(\frac{y-\mu}{\beta(0)}\right)
    \right\} dy\, dx. \nonumber
\end{align}
We split the inner integral at the symmetry point $y=\mu$:
\begin{align}
    \int_{-\infty}^\infty [\cdots]\, dy
    &= \int_{-\infty}^{\mu} [\cdots]\, dy
    + \int_{\mu}^{\infty} [\cdots]\, dy. \nonumber
\end{align}
By symmetry of $\bar{F}_1$, for any $\beta(x)>0$ and all $y<\mu$,
\begin{align}
    \bar{F}_1\!\left(\frac{y-\mu}{\beta(x)}\right)
    = 1 - \bar{F}_1\!\left(\frac{\mu-y}{\beta(x)}\right). \nonumber
\end{align}
Substituting into the lower-half integral, it becomes:
\begin{align}
    \int_{-\infty}^{\mu} [\cdots]\, dy
    &= -\int_{-\infty}^{\mu}
    \left\{
        \bar{F}_1\!\left(\frac{\mu-y}{\beta(x)}\right)
        - \bar{F}_1\!\left(\frac{\mu-y}{\beta(0)}\right)
    \right\} dy. \nonumber
\end{align}
Changing the integration limits by the reflection identity:
\[
\int_{-\infty}^{\mu} \phi(\mu - y)\, dy
= \int_{\mu}^{\infty} \phi(y-\mu)\, dy,
\]
we rewrite the lower-half integral as:
\begin{align}
    \int_{-\infty}^{\mu} [\cdots]\, dy
    &= -\int_{\mu}^{\infty}
    \left\{
        \bar{F}_1\!\left(\frac{y-\mu}{\beta(x)}\right)
        - \bar{F}_1\!\left(\frac{y-\mu}{\beta(0)}\right)
    \right\} dy. \nonumber
\end{align}
Hence, for each fixed $x$,
\[
\int_{-\infty}^\infty [\cdots]\, dy
= \int_{-\infty}^{\mu} [\cdots]\, dy
\ + \int_{\mu}^{\infty} [\cdots]\, dy
= -\int_{\mu}^{\infty}
    \left\{
        \bar{F}_1\!\left(\frac{y-\mu}{\beta(x)}\right)
        - \bar{F}_1\!\left(\frac{y-\mu}{\beta(0)}\right)
    \right\} dy \ + \int_{\mu}^{\infty}
    \left\{
        \bar{F}_1\!\left(\frac{y-\mu}{\beta(x)}\right)
        - \bar{F}_1\!\left(\frac{y-\mu}{\beta(0)}\right)
    \right\} dy = 0.
\]
Hence, ${Cov}(X, Y)
= \int_{0}^\infty \bar{F}_0(x) \cdot 0\, dx
= 0.$
\end{proof}

\begin{theorem} \label{thm: cov(x, y)}
Consider an accelerated conditionals model of the form: 
\[
P(X>x) = \bar{F}_0(x)
\quad \text{and} \quad
P(Y>y \mid X>x) = \bar{F}_1\left(\frac{y-\mu(x)}{\beta}\right), \quad x>0,\; y \in \mathbb{R},
\]
where $\bar{F}_0$ and $\bar{F}_1$ are survival functions and $\beta > 0$ and $\mu(x) \in \mathbb{R}$. Then:
\[
\mu'(x) \le 0 \quad \Longrightarrow \quad {Cov}(X,Y) \le 0,
\qquad
\mu'(x) \ge 0 \quad \Longrightarrow \quad {Cov}(X,Y) \ge 0.
\]
\end{theorem}

\begin{proof}
Since $P(Y>y , X> x) = \bar{F}_0(x)\bar{F}_1\left(\frac{y-\mu(x)}{\beta} \right)$ then $P(Y>y) = \lim_{x\to 0^+}P(Y>y , X> x) = \bar{F}_0(0)\bar{F}_1\left(\frac{y-\mu(0)}{\beta} \right) = \bar{F}_1\left(\frac{y-\mu(0)}{\beta}\right)$ since $\lim_{x\to 0^+}\bar{F}_0(x) = \bar{F}_0(0)  =  1$ and $\mu(0): = \lim_{x\to  0^+} \mu(x)$, we note:
\begin{align}
    Cov(X, Y) & = \int_{-\infty}^\infty \int_{0}^\infty\left( P(Y>y , X> x) - P(X> x)P(Y> y)   \right) \ dx\; dy \nonumber \\
    & = \int_{-\infty}^\infty \int_{0}^\infty  \bar{F}_0(x) \left( \bar{F}_1\left(\frac{y-\mu(x)}{\beta} \right)- \bar{F}_1 \left(\frac{y-\mu(0)}{\beta} \right) \right) \ dx \;dy.\nonumber 
\end{align}
Now since $\bar{F}_1$ is non-increasing (as $\bar{F}_1$ is a survival function), if $\mu'(x) \leq 0$ then $\bar{F}_1\left(\frac{y-\mu(x)}{\beta} \right) \leq \bar{F}_1\left(\frac{y-\mu(0)}{\beta} \right)$ and $Cov(X, Y) \leq 0$ for $x>0, y \in \mathbb{R}$. Conversely, if $\mu'(x) \ge 0$ then $\bar{F}_1\left(\frac{y-\mu(x)}{\beta} \right) \geq \bar{F}_1\left(\frac{y-\mu(0)}{\beta} \right)$ and $Cov(X, Y) \geq 0$ for $x>0, y \in \mathbb{R}$. Finally, if $\mu(x)$ is constant, say $\mu(x)\equiv c$, then $P(X>x,Y>y)=\bar F_0(x)\,\bar F_1\left(\frac{y -c}{\beta}\right)=P(X>x)\,P(Y>y)$, so $X$ and $Y$ are independent and ${Cov}(X,Y)=0$.
\end{proof}

\begin{theorem}\label{thm: shift identitiy}
Let $F_1$ be any (not necessarily symmetric) distribution function on $\mathbb{R}$ with survival $\bar F_1(u)=1-F_1(u)$, and let $\beta>0$, $\mu,\gamma\in\mathbb{R}$. Then
\begin{equation}\label{eq: shift identity}
\int_{-\infty}^{\infty}\Big(\bar F_1\!\big(\tfrac{y-\mu}{\beta}\big)-\bar F_1\!\big(\tfrac{y-\gamma}{\beta}\big)\Big)\,dy
\;=\; \mu-\gamma.
\end{equation}
\end{theorem}

\begin{proof}
Assume in addition that $F_1$ has a density $f_1$ so that $\bar F_1'(u)=-f_1(u)$. Define
\[
H(a):=\int_{-\infty}^{\infty}\bar F_1\!\Big(\tfrac{y-a}{\beta}\Big)\,dy,\qquad a\in\mathbb{R}.
\]
Formally differentiating under the integral sign (justified by dominated convergence since $0\le\bar F_1\le 1$ and $f_1$ integrates to $1$), we obtain
\[
H'(a)=\int_{-\infty}^{\infty}\frac{\partial}{\partial a}\bar F_1\!\Big(\tfrac{y-a}{\beta}\Big)\,dy
=\int_{-\infty}^{\infty}\frac{1}{\beta}f_1\!\Big(\tfrac{y-a}{\beta}\Big)\,dy
=\int_{-\infty}^{\infty} f_1(t)\,dt
=1,
\]
where we used the substitution $t=(y-a)/\beta$. Hence $H(a)=a+C$ for some constant $C$. Therefore,
\[
\int_{-\infty}^{\infty}\Big(\bar F_1\!\big(\tfrac{y-\mu}{\beta}\big)-\bar F_1\!\big(\tfrac{y-\gamma}{\beta}\big)\Big)\,dy
=H(\mu)-H(\gamma)=(\mu+C)-(\gamma+C)=\mu-\gamma,
\]
which is \eqref{eq: shift identity}.
\end{proof}

\section{Joint Densities and Log-likelihoods} \label{app: joint lik}
\subsection{Weibull-Logistic} 
We obtain the joint density function by differentiating the joint survival function \eqref{eq: logistic joint survival}:
\[
f_{X,Y}(x,y)
=
\frac{
\left(
\alpha^{\lambda} \beta \lambda x^{\lambda} e^{\frac{y}{\beta}} 
+ \alpha^{\lambda} \beta \lambda x^{\lambda} e^{\frac{\mu(x)}{\beta}} 
- x e^{\frac{y}{\beta}} \frac{d}{dx}\mu(x) 
+ x e^{\frac{\mu(x)}{\beta}} \frac{d}{dx}\mu(x)
\right)
e^{-\alpha^{\lambda} x^{\lambda} + \frac{y}{\beta} + \frac{\mu(x)}{\beta}}
}{
\beta^{2} x 
\left(
e^{\frac{3y}{\beta}} 
+ 3 e^{\frac{y + 2\mu(x)}{\beta}} 
+ 3 e^{\frac{2y + \mu(x)}{\beta}} 
+ e^{\frac{3\mu(x)}{\beta}}
\right)
},
\]
where $\lambda \in \mathbb{R}$ and $\alpha, \beta, \lambda > 0$ with $\tau \in [0, 1]$. The log-likelihood is thus:
\begin{align}
\ell(\alpha,\beta,\gamma,\lambda,\tau) 
&= \sum_{i=1}^n \Bigg[
\log\!\Big(
\alpha^{\lambda}\beta\lambda x_i^{\lambda}\big(e^{y_i/\beta}+e^{\mu(x_i)/\beta}\big)
+ x_i \frac{d\mu(x_i)}{dx}\big(-e^{y_i/\beta}+e^{\mu(x_i)/\beta}\big)
\Big)
- \alpha^{\lambda}x_i^{\lambda}
+ \frac{y_i}{\beta}
+ \frac{\mu(x_i)}{\beta} \nonumber\\
&\quad
- 2\log\beta
- \log x_i
- 3\log\!\big(e^{y_i/\beta}+e^{\mu(x_i)/\beta}\big)
\Bigg].\nonumber
\end{align}

\subsection{Weibull-Gumbel} 
We obtain the joint density function by differentiating the joint survival function \eqref{eq: gumbel joint survival}:\[
f_{XY}(x,y)
= 
\frac{
\left(
\alpha^{\lambda}\beta\lambda x^{\lambda} e^{\frac{y}{\beta}}
- x e^{\frac{y}{\beta}} \frac{d\mu(x)}{dx}
+ x e^{\frac{\mu(x)}{\beta}} \frac{d\mu(x)}{dx}
\right)
e^{-\alpha^{\lambda} x^{\lambda} - e^{\frac{-y+\mu(x)}{\beta}} - \frac{2y}{\beta} + \frac{\mu(x)}{\beta}}
}{
\beta^{2} x
},
\]
where $\lambda \in \mathbb{R}$ and $\alpha, \beta, \lambda > 0$ with $\tau \in [0, 1]$. The log-likelihood is thus:
\[
\ell(\alpha,\beta,\gamma,\lambda,\tau)
= \sum_{i=1}^n \left[
\log\!\left(
\alpha^{\lambda}\beta\lambda x_i^{\lambda} e^{\frac{y_i}{\beta}}
- x_i e^{\frac{y_i}{\beta}} \frac{d\mu(x_i)}{dx}
+ x_i e^{\frac{\mu(x_i)}{\beta}} \frac{d\mu(x_i)}{dx}
\right)
- \alpha^{\lambda} x_i^{\lambda}
- e^{\frac{-y_i+\mu(x_i)}{\beta}}
- \frac{2y_i}{\beta}
+ \frac{\mu(x_i)}{\beta}
- 2\log\beta
- \log x_i
\right].
\]

\subsection{Weibull-Laplace} 
We obtain the joint density function by differentiating the joint survival function \eqref{eq: laplace joint survival}:
\[
f_{XY}(x,y)
= \frac{
\left[\alpha^{\lambda}\beta\lambda x^{\lambda}
- \operatorname{sgn}\!\left(y-\mu(x)\right)\,x\,\mu'(x)\right]
\exp\!\left(-\alpha^{\lambda}x^{\lambda}
- \frac{\lvert y-\mu(x)\rvert}{\beta}\right)
}{2\beta^{2}x},
\]
where $\lambda \in \mathbb{R}$ and $\alpha, \beta, \lambda > 0$ with $\tau \in [0, 1]$. The log-likelihood is thus:
\[
\ell(\alpha,\beta,\gamma,\lambda,\tau)
=\sum_{i=1}^n\left[
\log\!\left(\alpha^{\lambda}\beta\lambda x_i^{\lambda}-\operatorname{sgn}\!\big(y_i-\mu(x_i)\big)\,x_i\,\mu'(x_i)\right)
-\alpha^{\lambda}x_i^{\lambda}
-\frac{\lvert y_i-\mu(x_i)\rvert}{\beta}
-\log 2
-2\log\beta
-\log x_i
\right].
\]
\subsection{Weibull-Cauchy} 
We obtain the joint density function by differentiating the joint survival function \eqref{eq: cauchy joint survival}:
\[
f_{XY}(x,y)
=
\frac{
\beta\left[
\lambda (\alpha x)^{\lambda}\!\left(\beta^{2} + (y - \mu(x))^{2}\right)
- 2x (y - \mu(x)) \frac{d\mu(x)}{dx}
\right]
e^{-(\alpha x)^{\lambda}}
}{
\pi x \left(\beta^{2} + (y - \mu(x))^{2}\right)^{2}
},
\]
where $\lambda \in \mathbb{R}$ and $\alpha, \beta, \lambda > 0$ with $\tau \in [0, 1]$. The log-likelihood is thus:
\[
\ell(\alpha,\beta,\gamma,\lambda,\tau)
=\sum_{i=1}^n \Bigg[
\log\!\Big(
\beta\!\left[
\lambda (\alpha x_i)^{\lambda}\!\left(\beta^{2} + (y_i - \mu(x_i))^{2}\right)
- 2x_i (y_i - \mu(x_i)) \frac{d\mu(x_i)}{dx}
\right]\!\Big) 
- (\alpha x_i)^{\lambda}
- \log\pi
- \log x_i
- 2\log\!\big(\beta^{2} + (y_i - \mu(x_i))^{2}\big)
\Bigg].
\]

\section{Metropolis-Hastings}\label{app: MH}
We detail the Metropolis-Hastings algorithm used in the study, to generate dependent pairs $(X,Y)$, where $f_{X, Y}$ can be obtained in Appendix \ref{app: joint lik}.\\

\noindent
\textbf{Step 1 (Initialization):} Choose any starting point \((x^{(0)},y^{(0)})\) in the support of \(f_{X,Y}\). Set \(j=0\). \\

\noindent
\textbf{Step 2 (Proposal):} Given the current state \((x^{(j)},y^{(j)})\), draw a candidate
\[
(x^*,y^*) \sim q(\,\cdot\,\mid x^{(j)},y^{(j)}).
\]
The study utilises an independence proposal:
\[
q(x^*,y^*\mid x^{(j)},y^{(j)}) = q(x^*,y^*)=f_X(x^*)\,f_Y(y^*).
\] \\
\noindent
\textbf{Step 3 (Acceptance ratio $\alpha_a$):} Compute:
\begin{align}
\alpha_a &= \min\left(
\frac{f_{X,Y}(x^*,y^*)\,q(x^{(j)},y^{(j)}\mid x^*,y^*)}{f_{X,Y}(x^{(j)},y^{(j)})\,q(x^*,y^*\mid x^{(j)},y^{(j)})}, 1 \right) \nonumber \\
& = \min \left(\frac{f_{X,Y}(x^*,y^*)}{f_X(x^*)\,f_Y(y^*)}
\bigg/
\frac{f_{X,Y}(x^{(j)},y^{(j)})}{f_X(x^{(j)})\,f_Y(y^{(j)})}, 1\right). \nonumber
\end{align}
Equivalently, using logs for numerical stability,
\begin{align}
\log(\alpha_a) 
&= \min\Bigg\{ 
\Big[\,\log\!\big(f_{X,Y}(x^*, y^*)\big)
      - \log\!\big(f_X(x^*)\big)
      - \log\!\big(f_Y(y^*)\big)\,\Big] \nonumber \\
&\quad\quad -
\Big[\,\log\!\big(f_{X,Y}(x^{(j)}, y^{(j)})\big)
      - \log\!\big(f_X(x^{(j)})\big)
      - \log\!\big(f_Y(y^{(j)})\big)\,\Big],
\, 0 \Bigg\}. \nonumber
\end{align}
Subsequently, $(x^*, y^*)$ is accepted as the new pair in the Markov chain under the following acceptance criterion:
\begin{align}
(x^{(j+1)}, y^{(j+1)}) =
\begin{cases}
(x^*, y^*), & \text{if } U < \alpha_a, \\
(x^{(j)}, y^{(j)}), & \text{otherwise.} 
\end{cases}\nonumber
\end{align}
where $U\sim Uniform(0, 1)$.

\section{Method of Moment Estimators} \label{app: mmes}
Suppose data in the form of $\textbf{X}^{(1)}, \textbf{X}^{(2)}, ..., \textbf{X}^{(n)}$ (where for each $i = 1, \ldots, n$, $\textbf{X}^{(i)} = (X_{1i}, X_{2i})^{T} )$ are obtained, which are independent and identically distributed according to \ref{eq: general marginal} and \ref{eq: general cond}. If $M_{1}, M_2 > 0$, consistent asymptotically normal method of moment estimates (m.m.e's) are acquired. Define:
\begin{eqnarray}
	M_{1} &=& \frac{1}{n}\sum_{i = 1}^{n} x_{1i}, \nonumber\\
	M_{2} &=& \frac{1}{n}\sum_{i = 1}^{n} x_{2i}, \nonumber\\
	S_{12}&=& \frac{1}{n}\sum_{i = 1}^{n}(x_{1i} - M_{1} ) ( x_{2i} - M_{2} ), \nonumber \\
    S_{1} &=&  \frac{1}{n}\sum_{i = 1}^{n}(x_{2i} - M_{2} )^2, \nonumber\\
	S_{2} &=&  \frac{1}{n}\sum_{i = 1}^{n}(x_{2i} - M_{2} )^2. \nonumber
\end{eqnarray}
We denote the method of moment estimators for $\alpha, \beta, \lambda, \gamma$ and $\tau$ as $\tilde{\alpha}, \tilde{\beta}, \tilde{\lambda},  \tilde{\gamma}$ and $\tilde{\tau}$ respectively. \\

Furthermore, $\tilde{\alpha}$ and $\tilde{\lambda}$ may be solved for numerically by equating \eqref{eq: Weibull ex} to $M_1$ and \eqref{eq: Weibull varx} to $S_1$. After solving for $\tilde{\alpha}$, $\tilde{\lambda}$ and $\tilde{\beta}$, use \eqref{eq: cov general} to obtain:
\begin{align}
    \tilde{\tau} = \left(\operatorname{sgn}\left(\mu'(x) \right)\frac{\tilde{\alpha}\tilde{\lambda} S_{12}}{\tilde{\beta}\Gamma\left(1+\frac{1}{\tilde{\lambda} } \right)}\right)^{1/\tilde{\lambda}}. \nonumber
\end{align}
\subsection{Weibull-Logistic}
Using \eqref{eq: ey logistic} and \eqref{eq: vary logistic}, we obtain:
\begin{align}
    \tilde{\gamma} &= {M_2}, \nonumber\\
    \tilde{\beta} & = \sqrt{3S_2/\pi^2}. \nonumber
\end{align}
\subsection{Weibull-Gumbel}
Using \eqref{eq: ey gumbel} and \eqref{eq: vary gumbel}, we obtain:
\begin{align}
    \tilde{\beta} & = \sqrt{6S_2/\pi^2}, \nonumber\\
    \tilde{\gamma} &= M_2- \tilde{\beta}\gamma^* . \nonumber
\end{align}
\subsection{Weibull-Laplace}
Using \eqref{eq: ey laplace} and \eqref{eq: vary laplace}, we obtain:
\begin{align}
    \tilde{\gamma} &= {M_2}, \nonumber\\
    \tilde{\beta} & = \sqrt{S_2/2}. \nonumber
\end{align}

\bibliographystyle{unsrtnat}
\bibliography{main}

\end{document}